\documentclass[12pt]{article}
\usepackage{makeidx}
\usepackage{color}
\usepackage{graphicx}
\usepackage{fullpage}
\usepackage{latexsym,amssymb}
\newtheorem{theorem}{Theorem}
\newtheorem{definition}{Definition}
\newenvironment{proof}{\noindent\textbf{Proof:}}{\hfill$\Box$}

\title{A Scalable Byzantine Grid}
\author{Alexandre Maurer \and S\'{e}bastien Tixeuil}
\date{
UPMC Sorbonne Universit\'{e}s, France
\\ \texttt{Alexandre.Maurer@lip6.fr}
\\ \texttt{Sebastien.Tixeuil@lip6.fr}
}

\begin{document}

\maketitle              

\begin{abstract}
Modern networks assemble an ever growing number of nodes.
However, it remains difficult to increase the number of channels per node, thus the maximal degree of the network may be bounded. This is typically the case in grid topology networks, where each node has at most four neighbors. In this paper, we address the following issue: if each node is likely to fail in an unpredictable manner, how can we preserve some global reliability guarantees when the number of nodes keeps increasing unboundedly ?

To be more specific, we consider the problem or reliably broadcasting information on an asynchronous grid in the presence of Byzantine failures -- that is, some nodes may have an arbitrary and potentially malicious behavior.
Our requirement is that a constant fraction of correct nodes remain able to achieve reliable communication.
Existing solutions can only tolerate a fixed number of Byzantine failures if they adopt a worst-case placement scheme. Besides, if we assume a constant Byzantine ratio (each node has the same probability to be Byzantine), the probability to have a fatal placement approaches 1 when the number of nodes increases, and reliability guarantees collapse.

In this paper, we propose the first broadcast protocol that overcomes these difficulties.
First, the number of Byzantine failures that can be tolerated (if they adopt the worst-case placement) now increases with the number of nodes.
Second, we are able to tolerate a constant Byzantine ratio, however large the grid may be. In other words, the grid becomes scalable. This result has important security applications in ultra-large networks, where each node has a given probability to misbehave.


\end{abstract}

\section{Introduction}

As modern networks grow larger and larger, their components become more likely to fail. Indeed, some nodes can be subject to crashes, attacks, bit flips, etc.
Many models of failures and attacks have been studied so far, but the most general one is the \emph{Byzantine} model \cite{LSP82j}: the failing nodes behave arbitrarily. In other words, we must anticipate the most malicious strategy they could adopt. This encompasses all other possible types of failures, and has important security applications.

In this paper, we study the problem of reliably broadcasting information in a network despite the presence of Byzantine failures. This is a difficult problem, as a single Byzantine node, if not neutralized, can potentially lie to the entire network. Our objective is to design a broadcast protocol that prevent or limit the diffusion of malicious messages.

\paragraph{Related works.}

Many Byzantine-robust protocols are based on \emph{cryptography} \cite{CL99c,DFS05c}: the nodes use digital signatures or certificates. Therefore, the correct nodes can verify the validity of received informations and authenticate the sender across multiple hops. However, this approach may not be as general as we want, as the malicious nodes are supposed to ignore some cryptographic secrets: therefore, their behavior is not \emph{completely} arbitrary.
Besides, cryptographic operations require the presence of a trusted infrastructure that deals with public and private keys: if this infrastructure fails, the whole network fails. Yet, we would like to consider that \emph{any} component can fail. For these reasons, we focus on cryptography-free solutions.

Cryptography-free solutions have first been studied in completely connected networks~\cite{LSP82j,AW98b,MMR03j,MRRS01c,MS03j}: a node can directly communicate with any other node, which implies the presence of a channel between each pair of nodes. Therefore, these approaches are hardly scalable, as the number of channels per node can be physically limited. We thus study solutions in multihop networks, where a node must rely on other nodes to broadcast informations.

A notable class of algorithms tolerates Byzantine failures with either space~\cite{MT07j,NA02c,SOM05c} or time~\cite{MT06cb,DMT11cb,DMT11j,DMT10cd,DMT10ca} locality. Yet, the emphasis of space local algorithms is on containing the fault as close to its source as possible. This is only applicable to the problems where the information from remote nodes is unimportant (such as vertex coloring, link coloring or dining philosophers). Also, time local algorithms presented so far can hold at most one Byzantine node and are not able to mask the effect of Byzantine actions. Thus, the local containment approach is not applicable to reliable broadcast.

It has been shown that, for agreement in the presence of up to $k$ Byzantine nodes, it is necessary and sufficient that the network is $(2k+1)$-connected, and that the number of nodes in the system is at least $3k+1$ \cite{D82j}. Also, this solution assumes that the topology is known to every node, and that nodes are scheduled according to the synchronous execution model.
Both requirements have been relaxed \cite{NT09j}: the topology is unknown and the scheduling is asynchronous. Yet, this solution retains $2k+1$ connectivity for reliable broadcast and $k+1$ connectivity for detection (the nodes are aware of the presence of a Byzantine failure). In sparse networks such as a grid (where a node has at most four neighbors), both approaches can cope only with a single Byzantine node, independently of the size of the grid. 

Another existing approach is based, not on connectivity, but on the fraction of Byzantine neighbors per node. Broadcast protocols have been proposed for nodes organized on a grid \cite{K04c,BV05c}. However, the wireless medium typically induces much more than four neighbors per node, otherwise the broadcast does not work. Both approaches are based on a local voting system, and perform correctly if every node has strictly less than a $1/4$ fraction of Byzantine neighbors. This result was later generalized to other topologies \cite{PP05j}, assuming that each node knows the global topology. Again, in weakly connected networks, this constraint on the proportion of Byzantine nodes in any neighborhood may be difficult to assess.

All aforementioned results rely on strong connectivity and Byzantine proportions assumptions in the network. In other words, tolerating more Byzantine failures requires to increase the number of channels per node, which may be difficult or impossible when the size of the network increases. 
To overcome this difficulty, an alternate approach has been proposed \cite{CtrZ}. The idea is to make a small concession to the problem: we now aim at reliable communication, not between \emph{all} correct nodes, but between \emph{most} correct nodes.
In other words, we now accept that a small minority of correct nodes can be fooled by the Byzantine nodes.
This is not unrealistic, as we already accepted the idea that some nodes can fail unpredictably (being hit by Byzantine failures).
This approach has been shown very efficient when the Byzantine failures are randomly distributed. This is the case, for instance, in a peer-to-peer overlay (the malicious nodes do not choose their localization when they join the overlay), or if we consider that each node has a given probability of failure.

All existing approaches have the same weak point: if the number of channels per node (degree) is bounded, a fixed number of Byzantine nodes can destabilize the whole network.
Indeed, if they adopt a sufficiently close formation, they can pretend to be the source node, and lie to \emph{any} other node -- thus, we cannot even ensure that \emph{most} correct nodes communicate reliably. Besides, if each node has a given probability to be Byzantine, the 
probability that such a fatal formation exists approaches $1$ when the number of nodes increases.
Therefore, these approaches are hardly scalable when the maximal degree is bounded.

\paragraph{Our contribution.}

In this paper, we propose the first broadcast protocol that overcomes these difficulties on a specific degree-bounded topology: the grid, where each node has at most four neighbors. For this protocol, the diameter of the grid can only have discrete values, but can be as large as we want. 
As in \cite{CtrZ}, our requirement is that a constant fraction of correct nodes achieves reliable communication. We show that the number of Byzantine failures that can be tolerated (if they adopt the worst-case placement) increases with the number of nodes: in other words, for the first time, this number is not limited by the maximal degree or the connectivity of the network.
Besides, if we assume a constant rate of Byzantine failures (each node has the same probability to be Byzantine), the expected reliable fraction of the network is always the same, however large the grid may be. This may have applications in large-scale networks, where each node has a given probability to fail: we can now increase the size of the network indefinitely, and yet preserve the same reliability guarantees.

The paper is organized as follows.
In Section~\ref{sec_A}, we describe the network topology (a sequence of grid networks that may be as large as we want) and the broadcast protocol to execute on it.
In Section~\ref{sec_B}, we adopt the point of view of an omniscient observer that knows the positions of Byzantine nodes, and give a methodology to determine a \emph{reliable node set} - that is, a set of nodes that always communicate reliably, in any possible execution.
At last, in Section~\ref{sec_C}, we use the aforementioned methodology to prove the claims.

\section{Our algorithm}
\label{sec_A}

In this section, we define a class of grid networks and the broadcast protocol to execute on.

\subsection{Hypotheses}

The network is constituted by a set of processes, called \emph{nodes}.
Some pairs of nodes are linked by a communication channel -- we call them \emph{neighbors} -- and can exchange messages.
Each node of the network has a unique identifier, which is its position on the grid.
A node, upon receiving a message from a neighbor, knows the identifier of this neighbor. The network is asynchronous: any message sent is eventually received, but it can be at any time.

\subsection{Network topology}

Let $N = 10$.
Our broadcast protocol is defined for the networks $G_k$, $\forall k \geq 1$,
$G_k$ being a $N^k \times N^k$ \emph{grid}.
These networks may be as large as needed.

\begin{definition}[Grid network]
\label{defgrid}
An $M \times M$ grid is a network such that:
\begin{itemize}
\item Each node has a unique identifier $(i,j)$ with
$0 \leq i < M$ and $0 \leq j < M$.
\item Two nodes $(i_{1},j_{1})$ and $(i_{2},j_{2})$ are neighbors if and only if one of these two conditions is satisfied:
\begin{itemize}
\item $i_{1} = i_{2}$ and $ | j_{1}-j_{2} |  = 1$.
\item $j_{1} = j_{2}$ and $ | i_{1}-i_{2} | = 1$.
\end{itemize}
\end{itemize}
\end{definition}

According to our hypotheses, each node knows its identifier $(i,j)$ on the grid, and the identifier $(i,j)$ of its neighbors.
Each node of $G_k$ also knows $N$ and $k$.

\subsection{Informal description of the protocol}

Our broadcast protocol (BP) is defined by induction: we use an existing BP on $G_1$, then use the BP of $G_k$ to define the BP of $G_{k+1}$. The idea is to associate a cluster of $G_{k+1}$ to each node of $G_k$.
Let $G(p)$ be the cluster associated to a node $p$ (we call it \emph{macro-node}).
This is illustrated in Figure~\ref{fig:cluster}.
The goal of a macro-node $G(p)$ is to simulate the behavior of $p$, so that we obtain a macroscopic BP in $G_{k+1}$. Then, when a node $u$ of $G(p)$ wants to broadcast a message $m$ in $G_{k+1}$:
\begin{enumerate}
\item First, $u$ broadcasts $m$ in $G(p)$ with a local BP.
\item Then, $G(p)$ broadcasts $m$ in $G_{k+1}$ with the macroscopic BP.
\end{enumerate}

\begin{figure*}
\begin{center}
\includegraphics[width=\textwidth]{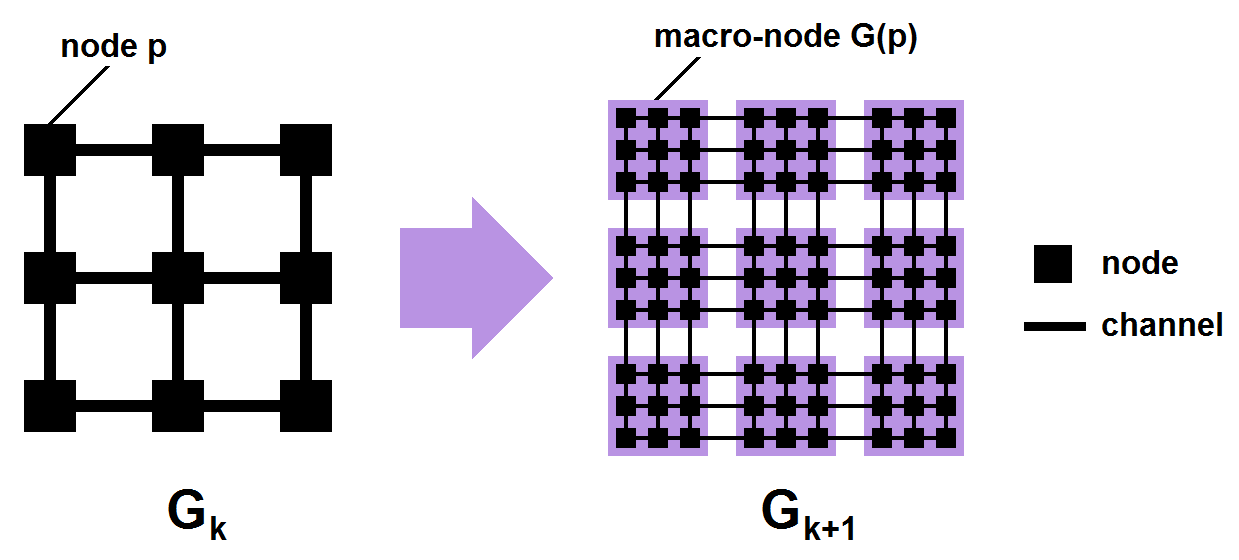}
\caption{Association of a macro-node of $G_{k+1}$ to each node of $G_k$} 
\label{fig:cluster}
\end{center}
\end{figure*}

The interest of this inductive definition lies in its Byzantine-resilience properties. These properties are studied in Section~\ref{sec_corr}.

\subsection{Complete description of the protocol}

\label{comp_desc}
The BP executed on $G_1$ is the \emph{Control Zone Protocol} (CZP) proposed in \cite{CtrZ}. Let us give the methodology to construct the BP of $G_{k+1}$ with the BP of $G_k$. For this purpose, we first give an algorithm to communicate between two macro-nodes (\emph{macro-channel}), then use it to construct the macroscopic BP.

\paragraph{Macro-node.}
To each node $p$ of $G_k$, we associate a cluster $G(p)$ of $G_{k+1}$, called macro-node.
Let $(i,j)$ be the identifier of $p$. Then, $G(p)$ is the $N \times N$ grid such that the node $(0,0)$ of $G(p)$ corresponds to the node $(Ni,Nj)$ of $G_{k+1}$.

\paragraph{Macro-channel.}
Let $p$ and $q$ be two neighbor nodes in $G_{k}$. We give an algorithm to tranfer messages from $G(p)$ to $G(q)$, as if they were two neighbor nodes linked by a channel.

First, we execute the $CZP$ on both $G(p)$ and $G(q)$, to enable local broadcast inside each macro-node.
The following algorithm enables to send a message $m$, known by the nodes of $G(p)$, to the nodes of $G(q)$.
Let $Border(p)$ (resp. $Border(q)$) be the set of nodes of $G(p)$ (resp. $G(q)$) having a neighbor in $G(q)$ (resp. $G(p)$).

\begin{enumerate}
\item The nodes of $Border(p)$ send $m$ to their neighbor in $Border(q)$.
\item The nodes of $Border(q)$, upon receiving $m$ from their neighbor in $Border(q)$, broadcast $m$ in $G(q)$ with the $CZP$.
\item The nodes of $G(q)$, upon receiving strictly more than $N/2$ distinct messages $(v_i,m)$ trough the $CZP$ with $v_i \in Border(q)$, accept $m$.
\end{enumerate}

We associate a dynamic set $Sen_q$ to each node of $G(p)$ (storing the message to send), and a dynamic set $Rec_p$ to each node of $G(q)$ (storing the messages received).
We execute this algorithm for each pair of neighbor macro-nodes. This mechanism is illustrated in Figure~\ref{fig:passing}.

\paragraph{Macroscopic BP.} For each node $p$ of $G_k$, all nodes of $G(p)$ execute the same algorithm than $p$, with the two following modifications:
\begin{enumerate}
\item When the algorithm requires to send a message $m$ to a neighbor $q$, add $m$ to $Sen_q$.
\item When a message $m$ is added to the set $Rec_q$, consider that $m$ was received from $q$.
\end{enumerate}
Now, let $s$ be a node of $G(p)$ that wants to broadcast a message $m$ in $G_{k+1}$. First, $s$ broadcasts $(s,m)$ in $G(p)$ with the CZP. Then, upon receiving $(s,m)$, the nodes of $G(p)$ broadcast $(s,m)$ with the macroscopic $BP$. Thus, the nodes receiving $(s,m)$ know that $s$ broadcast $m$: we now have a BP on $G_{k+1}$.

\begin{figure*}
\begin{center}
\includegraphics[width=\textwidth]{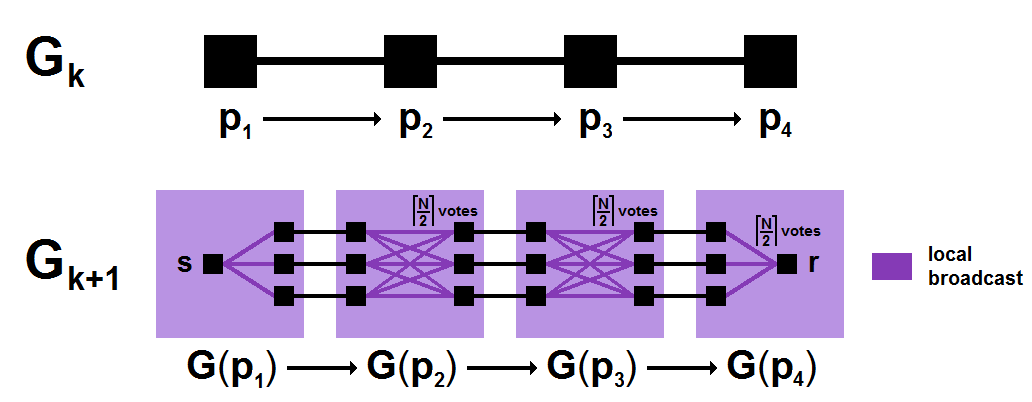}
\caption{Principle of the protocol} 
\label{fig:passing}
\end{center}
\end{figure*}

\section{Construction of a reliable node set}
\label{sec_B}
\label{sec_corr}

In this section, we now assume that some nodes are Byzantine, and behave arbitrarily instead of following the aforementioned protocol.
We adopt the point of view of an omniscient external observer, knowing the positions of Byzantine nodes, and give a methodology to determine a \emph{reliable node set} - that is, a set of nodes that communicate reliably in any possible execution.
This methodology is used in Section~\ref{sec_C} to prove the claims.
Notice that we never require that a node determines such a set: this is just a global view of the system.

\paragraph{Notion of reliable node set.}

The nodes following the aforementioned protocol are called \emph{correct}. The correct nodes do not know the positions of Byzantine nodes.

\begin{definition} [Reliable node set]
\label{def_rns}
For a given broadcast protocol (BP), a set of correct nodes is \emph{reliable} if, for each pair of nodes $s$ and $r$ of this set:
\begin{enumerate}
\item If $s$ broadcasts $m$, $r$ eventually accepts $(s,m)$.
\item If $r$ accepts $(s,m)$, $r$ necessarily broadcast $m$.
\end{enumerate}
\end{definition}

In other words, a reliable node set behaves like a network without Byzantine failures.
The item (1) guarantees that the nodes always manage to communicate.
The item (2) guarantees that no node of the reliable set can be fooled - for instance, if a Byzantine node broadcasts $(s,m')$ to make the network believe that $s$ broadcast $m'$.

\paragraph{Construction of a reliable node set.}

Let $Corr$ be a set of correct nodes of $G_k$. Let us define a function $Rel_k$ such that $Rel_k(Corr)$ returns a reliable node set for our BP. For this purpose, we first introduce some new elements.

In \cite{CtrZ}, we gave a methodology to determine a reliable node set for the CZP on an $N \times N$ grid, for a given set $Corr_0$ of correct nodes. Let $Rel_{CZP}$ be a function such that $Rel_{CZP}(Corr_0)$ returns a reliable node set for the CZP.

At last, we introduce the notion of \emph{correct macro-node}.
In broad outline, a correct macro-node behaves like a correct node in the macroscopic BP. This intuitive idea is the key element of the next theorem.

\begin{definition}[Correct macro-node]
\label{def_macrocorr}
Let there be an $N \times N$ grid with a distribution $Corr_0$ of correct nodes.
This grid (or macro-node) is said \emph{correct} if each side of the grid (up, down, right and left), among its $N$ nodes, has strictly more than $3N/4$ nodes in $Rel_{CZP}(Corr_0)$.
\end{definition}

The underlying idea of this definition is the following: the reliable node sets of two adjacent correct macro-nodes are always connected by a majority of channels (strictly more than $N/2$). Therefore, the messages exchanged between these two reliable sets always receive a majority of votes. This idea is illustrated in Figure~\ref{fig:macrocorr}, and used in the proof below.

\begin{figure*}
\begin{center}
\includegraphics[width=\textwidth]{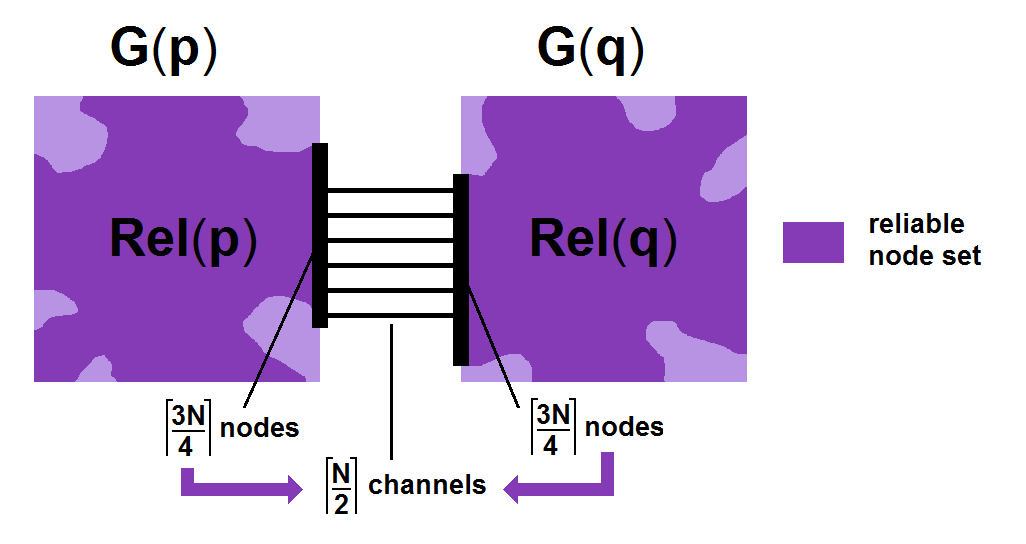}
\caption{Reliable communication between $2$ correct macro-nodes} 
\label{fig:macrocorr}
\end{center}
\end{figure*}

\newpage

We can now define the function $Rel_k$ by induction, $\forall k \geq 1$:
\begin{itemize}
\item $Rel_1 = Rel_{CZP}$ 
\item $ Rel_{k+1}(Corr) = \bigcup_{p \in Rel_k(Corr')} Rel_{CZP}(Corr(p))$, where \dots 
\begin{itemize}
\item $Corr$ is a distribution of correct nodes on $G_{k+1}$.
\item $Corr(p)$ is the corresponding distribution on $G(p)$.
\item $Corr'$ is the set of nodes $p$ of $G_k$ such that $G(p)$ is a correct macro-node.
\end{itemize}
\end{itemize}
In the following, we refer to $Rel_{CZP}(Corr(x))$ by $Rel(x)$.

\begin{theorem}
\label{thmain}
$\forall k \geq 1$, if $Corr$ is a distribution of correct nodes on $G_k$, then $Rel_k(Corr)$ is a reliable node set for our BP.
\end{theorem}

\begin{proof} The main idea of the proof is to show an equivalence between the execution on $G_{k+1}$ and a virtual execution on $G_k$ (this, of course, does not mean that $G_k$ must actually exist for $G_{k+1}$ to work).

The proof is by induction. The property is true at rank $1$ by definition. Now, let us suppose that the property is true at rank $k$, and show that it is true at rank $k+1$.
Let $Corr$ be a distribution of correct nodes on $G_{k+1}$, and let $s$ and $r$ be two nodes of $Rel_{k+1}(Corr)$. Let us suppose that $s$ broadcasts $m$ in $G_{k+1}$. Then, to show that $Rel_{k+1}(Corr)$ is a reliable node set, we show that the items (1) and (2) of Definition~\ref{def_rns} are satisfied.

\begin{enumerate}
\item We call \emph{accumulative} a distributed algorithm where each node holds a given number of dynamic sets $S_1, S_2, S_3 \dots$, can only add elements to these sets ($S_i \leftarrow S_i \cup \{x\}$), and eventually executes an action when a given collection of elements has joined these sets: $(X_1 \subseteq S_1) \land (X_2 \subseteq S_2) \land \dots$. The CZP is accumulative, and so is our BP, as it is an inductive combination of accumulative algorithms. In other words, the order of reception of messages is unimportant in our BP.

Let $p$ and $q$ be the nodes of $G_k$ such that $s$ belongs to $G(p)$ and $r$ belongs to $G(q)$. By definition of $Rel_{k+1}$, $p$ and $q$ belong to $Rel_k(Corr')$. Let us suppose that $Corr'$ is a distribution of correct nodes on $G_k$. Then, $Rel_k(Corr')$ is a reliable node set on $G_k$. Therefore, if $p$ broadcasts $(s,m)$, there exists a sequence of message receptions such that $q$ eventually accepts $(s,m)$.
Let $(R_1,R_2,\dots,R_M)$ be this sequence, $R_i$ being a triplet $(q_i,m_i,p_i)$ such that $q_i$ receives $m_i$ from $p_i$, with $p_1 = p$ and $q_M = q$. Let us prove the following property $\mathcal{P}_i$ by induction, $\forall i \in \{1,\dots,M\}$: all the nodes of $Rel(q_i)$ eventually add $m_i$ to $Rec_{p_i}$.

\begin{itemize}

\item First, let us show that $\mathcal{P}_1$ is true. According to our BP, $s$ initially broadcasts $(s,m)$ in $G(p)$.
Therefore, as $p = p_1$, all the nodes of $Rel(p_1)$ eventually accept $(s,m)$. Then, as they execute the same alogorithm than $p_1$, they add $m_1$ to their set $Sen_{q_1}$.

Let $Border(q_1)$ be the set of nodes of $G(q_1)$ having a neighbor in $G(p_1)$.
As $G(q_1)$ and $G(p_1)$ are two correct macro-nodes, according to Definition~\ref{def_macrocorr}, strictly more than $N/2$ nodes of
$Rel(p_1)$ have a neighbor in $Rel(q_1)$.
Therefore, strictly more than $N/2$ nodes of $Border(q_1) \cap Rel(q_1)$ eventually receive $m_1$, and broadcast it in $G(q_1)$.
So all the nodes of $Rel(q_1)$ eventually receive strictly more than $N/2$ messages
$(v_x,m_1)$ with $v_x \in Border(q_1)$ and add $m_1$ to $Rec_{p_1}$.
Thus, $\mathcal{P}_1$ is true.

\item Now, let us suppose that $\mathcal{P}_j$ is true $\forall j \leq i$.
Then, as the order of reception of messages is unimportant,
all the nodes of $Rel(p_{i+1})$ eventually behave as $p_{i+1}$, and add $m_{i+1}$ to $Sen_{q_{i+1}}$.

Thus, by a perfectly similar demonstration, $\mathcal{P}_{i+1}$ is true.

\end{itemize}

Then , as $r \in Rel(q)$, according to $\mathcal{P}_M$: $r$ eventually receives the same messages as $q = q_M$ and accepts $(s,m)$. Thus, the item (1) of Definition~\ref{def_rns} is satisfied. This is illustrated in Figure~\ref{fig:proof1}.

\begin{figure*}
\begin{center}
\includegraphics[width=\textwidth]{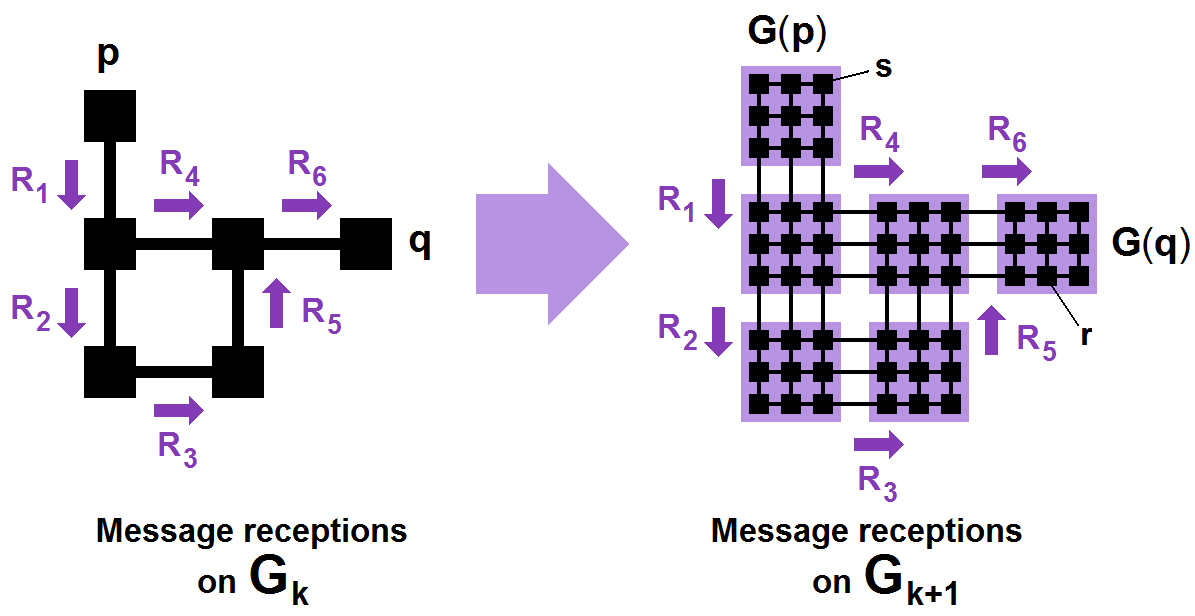}
\caption{Illustration of the proof (1) : what occurs in $Rel_k(Corr')$ eventually occurs in $Rel_{k+1}(Corr)$} 
\label{fig:proof1}
\end{center}
\end{figure*}

\item The proof is by contradiction. Let us suppose the opposite: $r$ accepts a message $(s,m)$, yet $s$ did not broadcast $m$. Let $p_0$ be the node of $G_k$ such that $r \in Rel(p_0)$. If we also have $s \in Rel(p_0)$, it is impossible that $r$ accepts $(s,m)$, as $Rel(p_0)$ is a reliable node set.
So $s$ necessarily belongs to another macro-node.
Similarly than above, let us suppose that $Corr'$ is a distribution of correct nodes on $G_k$.
Then, as $Rel_k(Corr')$ is a reliable node set on $G_k$, $r$ necessarily received a message that $p_0$ cannot receive in $G_k$. Let us show that this is impossible.

Let $u$ be the first node of $Rel_{k+1}(Corr)$ (possibly $r$), belonging to a macro-node $G(q)$, to receive a message $m'$ that $q$ cannot receive in $G_k$. Let $G(p)$ be the macro-node sending this message.
If $G(p)$ is not correct (in the sense of Definition~\ref{def_macrocorr}), then $p$ does not belong to $Corr'$, is assumed to be Byzantine on $G_k$, and can actually send $m'$ to $q$ -- so $G(p)$ is necessarily correct. 
It implies that $u$ received strictly more than $N/2$ messages $(v_i,m')$ with $v_i \in Border(q)$. As $G(p)$ and $G(q)$ are two correct macro-node, strictly more than $N/2$ nodes of $Rel(p)$ have a neighbor in $Rel(q)$. So at least one of the nodes $v_i$ belongs to $Rel(q)$ \emph{and} received $m'$ from a neighbor $v \in Rel(p)$.
As $Rel(p)$ is a reliable node set, the only possibility is that $v$ received a message that $p$ cannot receive in $G_k$. So $u$ is not the first node in this situation, which contradicts the initial statement.
Thus, the item (2) of Definition~\ref{def_rns} is satisfied. This is illustrated in Figure~\ref{fig:contrad}.

\begin{figure*}
\begin{center}
\includegraphics[width=\textwidth]{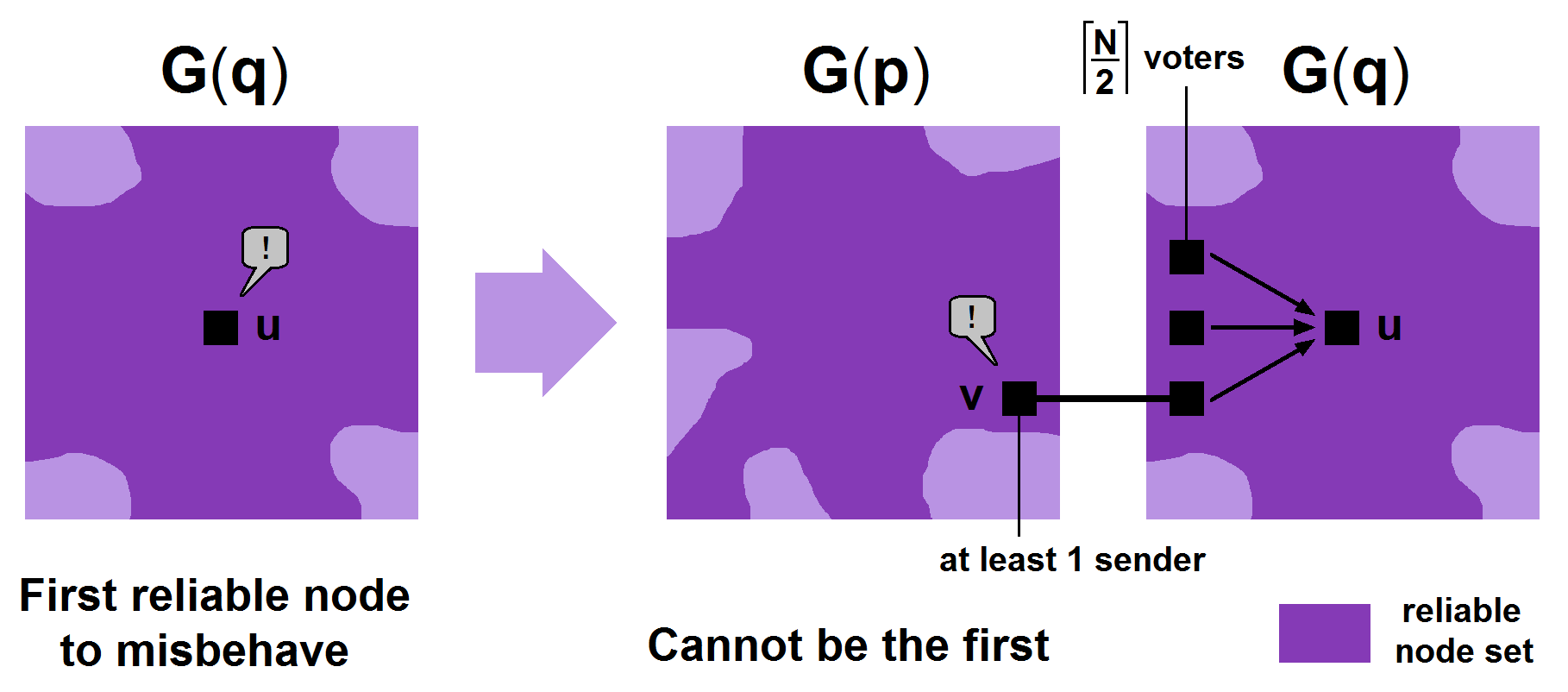}
\caption{Illustration of the proof (2) : a node of $Rel_{k+1}(Corr)$ cannot misbehave} 
\label{fig:contrad}
\end{center}
\end{figure*}

\end{enumerate}

\end{proof}

We now have a methodology to determine a reliable node set for a given distribution of Byzantine nodes on $G_k$, $\forall k \geq 1$. In the next section, we use this methodology to prove the claims.

\section{Proof of the claims}
\label{sec_C}

In this section, we finally prove the claims of the paper: the number of Byzantine failures that can be tolerated increases with the number of nodes (if they adopt the worst-case placement), and a constant rate of Byzantine failures can be tolerated, however large the grid may be.
As in \cite{CtrZ}, our requirement to \emph{tolerate} Byzantine failures is that a constant fraction of the network communicates reliably.

\subsection{Worst-case placement}

Let us give a minimal number of Byzantine failures that can be tolerated when they adopt an arbitrary placement (possibly the worst).

\begin{theorem}
$\forall k \geq 1$, on a grid $G_k$ with at most $2^{k-1}$ Byzantine failures (arbitrarily placed), the fraction of the network achieving reliable communication is at least $\displaystyle  1 - \frac{4}{N^2}$.
\end{theorem}

\begin{proof}
The proof is by induction. For $k = 1$, we can test all possible placements of a single Byzantine failure (as $N = 10$) and show that the property is true. Now, let us suppose that the property is true at rank $k$. Let there be $2^{k}$ Byzantine failures arbitrarily placed on $G_{k+1}$. Then, at most $2^{k-1}$ macro-nodes of $G_{k+1}$ contain more than $2$ Byzantine failures.
Again, by testing all possible cases, we can show that an $N \times N$ grid with at most $1$ Byzantine failure is always correct in the sense of Definition~\ref{def_macrocorr}.
So at most $2^{k-1}$ macro-nodes are not correct. Therefore, as the property is true at rank $k$, the reliable node set covers at least a $\displaystyle  1 - \frac{4}{N^2}$ fraction of macro-nodes (and in this worst case, all these macro-nodes have only correct nodes). Thus, according to the definition of $Rel_{k+1}$, the property is true at rank $k+1$. This is illustrated in Figure~\ref{fig:worst}.

\begin{figure*}
\begin{center}
\includegraphics[width=\textwidth]{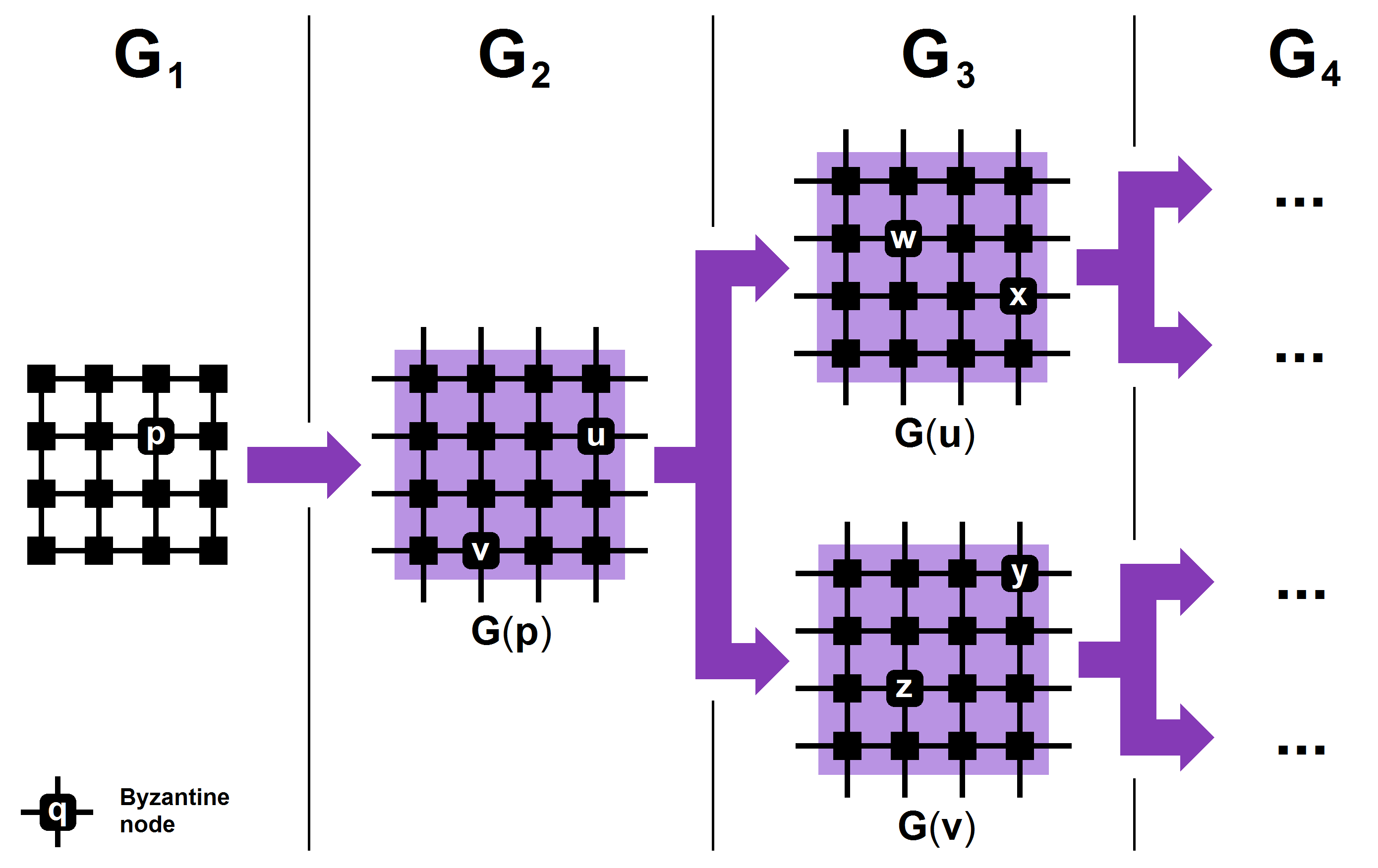}
\caption{Worst-case placement of $2^{k-1}$ Byzantine nodes on $G_k$} 
\label{fig:worst}
\end{center}
\end{figure*}

\end{proof}

So we can always tolerate $2^{k-1}$ failures on $G_k$. As the parameter $k$ sets the size of the grid, this number increases with the number of nodes.
To our knowledge, this is the first time that this number is not limited by the connectivity or the maximal degree of the network.

\subsection{Random distribution}

Let us assume a constant rate of Byzantine failures (each node has the same probability $\lambda$ to be Byzantine) and give the expected reliable fraction of the network.
Let $\mu = 1 - \lambda$ be the probability that a node is correct.

\begin{theorem} 
$\forall k \geq 1$, let $F_k(\mu)$ be the expected reliable fraction of $G_k$.
Then, if $\mu \geq 1 - 10^{-5}$, we have $ F_k(\mu) \geq 1 - 10^{-4}$.
\end{theorem}

\begin{proof}
Let there be an $N \times N$ grid where each node has the same probability $\mu_0$ to be correct. We call $P(\mu_0)$ the probability that the two following events occur:
\begin{enumerate}
\item The grid is \emph{correct} in the sense of Definition~\ref{def_macrocorr}.
\item A node, chosen uniformly at random, belongs to $Rel_{CZP}(Corr_0)$, $Corr_0$ being the distribution of correct nodes on the grid.
\end{enumerate}
We want to prove the following property by induction: $\displaystyle  F_k \geq \prod_{i=1}^{i=k} P^i(\mu)$, $P^i$ being the $i^{th}$ application of the function $P$. The property is true at rank $1$, as $F_1(\mu) \geq P(\mu)$.

Now, let us suppose that the property is true at rank $k$. Let $Corr$ be the distribution of correct nodes on $G_{k+1}$. Let $u$ be a randomly chosen node of $G_{k+1}$, and let $p$ be the node of $G_k$ such that $u$ belongs to the macro-node $G(p)$. According to Theorem~\ref{thmain}, to have $u \in Rel_{k+1}(Corr)$, it is necessary and sufficient that (1) $u \in Rel(p)$  and (2) $p \in Rel_k(Corr')$.
The first event occurs with probability $P_1 \geq P(\mu)$, and if so, the second event occurs with probability $P_2 \geq F_k(P(\mu))$.
Thus, $\displaystyle F_{k+1}(\mu) \geq P(\mu)F_k(P(\mu)) = \prod_{i=1}^{i=k+1} P^i(\mu)$: the property is true at rank $k+1$.
This is illustrated in Figure~\ref{fig:suff}.

\begin{figure*}
\begin{center}
\includegraphics[width=\textwidth]{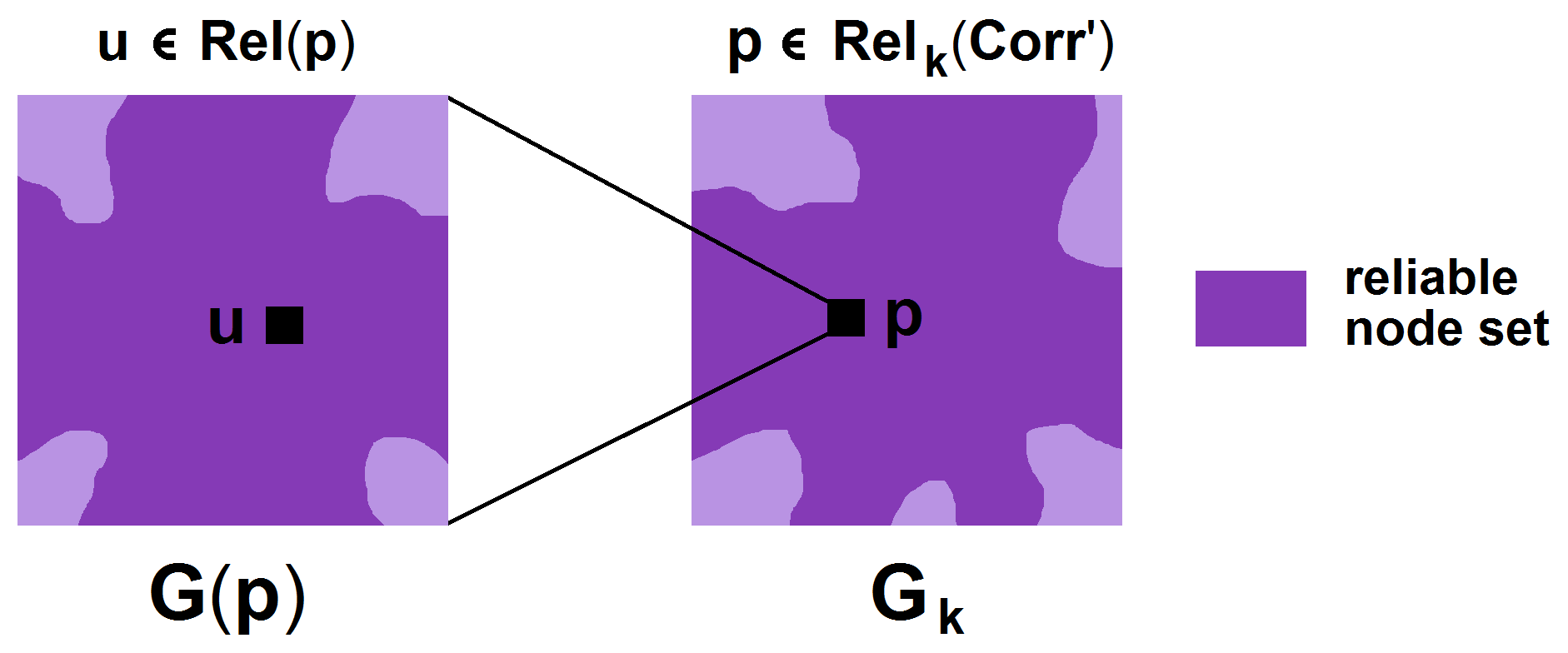}
\caption{Sufficient condition for $u$ to be in $Rel_{k+1}(Corr)$} 
\label{fig:suff}
\end{center}
\end{figure*}

\paragraph{}
Now, let us give a lower bound of $P(\mu_0)$. We consider two disjoint cases:

\begin{enumerate}

\item The case where all the nodes of the $N \times N$ grid are correct, which occurs with probability $\mu_0^{N^2}$. In this case, $Rel_{CZP}(Corr_0)$ covers the whole grid, and the grid is 
correct in the sense of Definition~\ref{def_macrocorr}.

\item The case where one single node is Byzantine, which occurs with probability 
$N^2(1-\mu_0) \mu_0^{N^2 - 1}$. As $N = 10$, we evaluate $Rel_{CZP}(Corr_0)$ for the $100$ possible placements of the single Byzantine node.
In $64$ cases, this set contains $99$ nodes.
In $32$ cases, it contains $98$ nodes.
In $4$ cases, it contains $96$ nodes.
Thus, the probability that a randomly chosen correct node belongs to this set is
$\displaystyle  \alpha = \frac{64 \times 99 + 32 \times 98 + 4 \times 96}{100 \times 99} \geq \frac{199}{200}$. In all cases, the grid is correct in the sense of Definition~\ref{def_macrocorr}. This is illustrated in Figure~\ref{fig:cases}.

\end{enumerate}

\begin{figure*}
\begin{center}
\includegraphics[width=\textwidth]{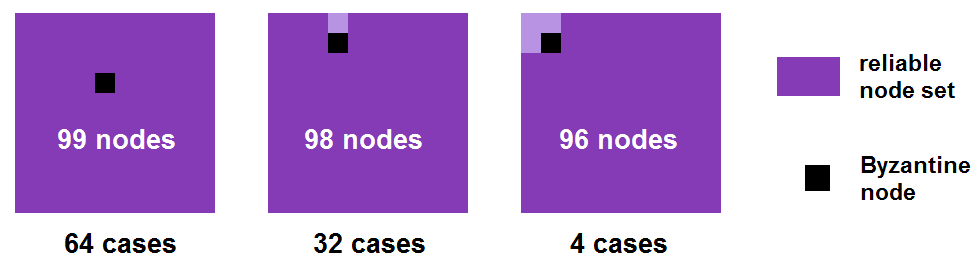}
\caption{Different cases for the placement of $1$ Byzantine node on an $N \times N$ grid} 
\label{fig:cases}
\end{center}
\end{figure*}

So $\displaystyle P(\mu) \geq g(\mu) = \mu^{N^2} + \alpha N^2(1-\mu) \mu^{N^2 - 1}$. This function is convex $\displaystyle \left( \frac{\partial^2 g(\mu)}{\partial \mu^2} \leq 0 \right)$ for $\mu \geq \alpha$.
Let $\beta = 1 - 10^{-5} \geq \alpha$.
Then, $\forall \mu \geq \beta$,
$g(\mu) \geq f(\gamma,\mu) = 1 - \gamma(1 - \mu)$,
with $\displaystyle \gamma = \frac{1-g(\beta)}{1-\beta}$.
Then, we easily show by induction that $\forall k \geq 1$, $P^k(\mu) \geq f(\gamma^k,\mu)$. So $\displaystyle F_{k}(\mu) \geq H_k(\mu) = \prod_{i=1}^{i=k} f(\gamma^i,\mu)$.

\paragraph{}

We now have a lower bound of $F_k(\mu)$, but it may be hard to calculate when $k$ approaches infinity. To overcome this difficulty, let $i_0$ be the first integer such that, $\forall i \geq i_0$, $\displaystyle \gamma^i \leq \frac{1}{i^2}$.
So $\displaystyle H_{k}(\mu) \geq \prod_{i=1}^{i=i_0} f(\gamma^i,\mu)
\prod_{i=i_0+1}^{i=k} (1 - \frac{1-\mu}{i^2})$.
Then, when $k$ approaches infinity, we can apply the Wallis formula: $\displaystyle \lim\limits_{x \to \infty} H_k(\mu) \geq
\prod_{i=1}^{i=i_0} f(\gamma^i,\mu)
\frac{sin(\pi \sqrt{1-\mu})}{\pi \sqrt{1-\mu}}
\geq
1 - 10^{-4}$
if $\mu \geq \beta$. Thus, the result, as $H_k(\mu)$ decreases with $k$.
\end{proof}

Therefore, we can hold a constant rate of Byzantine failures and yet have a constant expected fraction of reliable nodes, however large the grid may be. This may have important security applications -- for instance in a computationnal grid where each processor has a given probability to misbehave. This result shows that, for a given security requirement, we can increase the size of the grid indefinetely, which could be a solution to the problem of scalability.

\section{Conclusion}

In this paper, we have shown that Byzantine resilience was possible in a scalable degree-bounded network.
If the adversary can place the Byzantine nodes arbitrarily, then for the first time, we can tolerate a number of Byzantine failures that largely exceeds the node degree. If not (random distribution), then we can tolerate a constant fraction of Byzantine nodes, even if the size of the network approaches infinity.

We have the strong conviction that this approach (slice the network into clusters, then slice each cluster into smaller clusters, etc \dots) can be generalized to less regular topologies.
Indeed, the notion of a correct macro-node (see Definition~\ref{def_macrocorr}) can be generalized to an arbitrary graph -- the key idea is that, for each interface with another macro-node, we must still have a $3/4$ fraction of reliable nodes. Besides, the network diameter can only have discrete values here, but we could generalize the result to any network diameter.

\bibliographystyle{plain}
\bibliography{biblio}

\end{document}